\documentclass[twocolumn,pra,superscriptaddress,nofootinbib]{revtex4-1}

\usepackage{makeidx}
\makeindex

\usepackage{graphicx}
\usepackage{changes}
\usepackage{amssymb}

\usepackage{amscd}

\usepackage{tikz}
\usetikzlibrary{chains}
\usetikzlibrary{fit}
\usepackage{pgflibraryarrows}		
\usepackage{pgflibrarysnakes}		
\usepackage{xcolor}
\usepackage{epsfig}

\usepackage{amsthm}
\usepackage{amsmath}
\usepackage{amsfonts}
\usepackage{amssymb,amstext}
\usepackage{mathtools}

\usepackage[justification=centering]{caption}

\usepackage[colorlinks=true,urlcolor=blue, hyperindex] {hyperref}

\usepackage{etoolbox}

\theoremstyle{plain}

\newtheorem{mylem}{Lemma}

\theoremstyle{definition}
\newtheorem{mydef}{Definition}

\newtheorem{mythm}{Theorem}

\newcommand{\id}{\mathrm{id}}
\newcommand{\ket}[1]{\left| #1 \right \rangle}
\newcommand{\bra}[1]{\left \langle #1 \right|}

\newcommand{\eps}{\varepsilon}

\newcommand{\sysC}{\mathrm{C}}

\newcommand{\CtoCT}{\mathrm{C}\rightarrow\mathrm{CT}}

\newcommand{\map}{\mathcal{M}}

\DeclareMathOperator{\Pos}{Pos}

\newcommand*{\cC}{\mathcal{C}}

\newcommand*{\cI}{\mathcal{I}}
\newcommand*{\cM}{\mathcal{M}}

\DeclareMathOperator{\tr}{tr}

\definecolor{darkgreen}{rgb}{0,0.4,0}
\definecolor{orange}{rgb}{1,0.5,0}

\newcommand{\proj}[1]{\ket{#1}\!\!\bra{#1}}

\begin{document}

\title{Performance of stochastic clocks in the Alternate Ticks Game}
\author{Sandra Stupar}
\email{rankovic@itp.phys.ethz.ch}
\affiliation
{Institute for Theoretical Physics, ETH Zurich, 8093 Zurich, Switzerland}
\author{Christian Klumpp}
\affiliation
{Institute for Theoretical Physics, ETH Zurich, 8093 Zurich, Switzerland}
\author{Nicolas Gisin}
\affiliation
{Group of Applied Physics, University of Geneva, 1211 Geneva, Switzerland}
\author{Renato Renner}
\affiliation
{Institute for Theoretical Physics, ETH Zurich, 8093 Zurich, Switzerland}

\begin{abstract}
In standard quantum theory, time is not an observable. It enters as a parameter in the Schr\"odinger equation, but there is no measurement operator associated to it.  Nevertheless, one may take an operational viewpoint and regard time as the information one can read from clocks. The {\it Alternate Ticks Game}, introduced in [arXiv:1506.01373], is a completely operational means to quantify the accuracy of time scales generated by clocks. The idea is to count the number of ticks that two copies of a clock can produce until they run out of synchronisation. Here we investigate the performance of {\it stochastic clocks} in this game. These are clocks which are classical in the sense that they do not exploit quantum coherence. Our results support earlier conjectures that their accuracy grows  linearly in the size of the clockwork, measured in terms of the dimension of the associated Hilbert space. In particular, we derive explicit bounds on the accuracy of a natural class of stochastic clocks, the {\it stochastic ladder clocks}. 
\end{abstract}

\maketitle

Time enters as a parameter in the quantum-mechanical equations of motion. Treating time in this way is convenient in many cases, but it also bears conceptual problems. These become apparent  if one tries to reconcile quantum mechanics with other physical theories. Indeed, one talks about the ``problem of time'' when describing obstacles that arise in attempts to quantise general relativity~\cite{DeWitt,Kuchar,Isham}. The treatment of time as a parameter is also problematic from a thermodynamic point of view~\cite{ThermodynamicsOfTime}. An alternative could be to represent time in terms of an operator, analogously to the quantum-mechanical position or momentum operators, but this turned out to be impossible in general~\cite{Pauli,PauliTwo, BuschTimeOperator,PashbyThesis}. 

These problems may be avoided by an operational approach to time. According to this approach,  time is what one reads from a clock. In other words, time scales are generated by  clocks, which are modelled as physical systems. The approach has a long tradition and certainly played a key role in the development of relativity theory~\cite{EinsteinClock}. It is therefore to be expected that it can also teach us something about quantum mechanics and its relation to other physical theories, as well as about cosmology. For example, the proposal to describe time as correlations between a clock and other systems, and extensions thereof,~\cite{PageWootters,Wootters,Mott,Briggs,Brout,GiovannettiQuantumTime,ChiaraVlatko} certainly fits into this operational picture. The operational approach has also inspired work in quantum information theory. An example is the question whether clock synchronisation schemes that use quantum communication have advantages over classical ones~\cite{PreskillClockSynchronisation,JozsaClockSynchronisation, Eddington}. Another one, related to quantum thermodynamics, is whether measuring time has a thermodynamical cost~\cite{PauliThesis, MischaRalph, ThermodynamicQClock, InflationUniverse}.

Clocks, regarded as actual physical systems, are obviously subject to  physical constraints, which may limit their accuracy. One of these constraints is that clocks have a finite size. Taking a quantum-information perspective, it is natural to quantify size in terms of the dimension~$d$ of the Hilbert space associated to the clockwork. This is made precise by the framework proposed in~\cite{ClockModel}, where clocks are modelled as general quantum systems that emit time information, the \emph{ticks}, to their environment.  

The accuracy of such tick-emitting clocks can be quantified in a purely operational manner in terms of the \emph{Alternate Ticks Game (AT~Game)}~\cite{ClockModel}. In this game, the aim is to set up two clocks in such a way that their ticks occur in an alternating order for as long as possible. The game stops as soon as one of the clocks emits two ticks in a row without a tick by the other clock in between. The score, $N$, of the game is the expected number of ticks emitted by each clock until this happens.  The score $N$ achieved by the best strategy is then a measure for how well the two clocks can be synchronised. To turn this into a measure for the \emph{accuracy} of a clock, one simply plays the AT~Game with two copies of the clock.

Here we focus on a particular type of clocks, called \emph{stochastic clocks}. These are, in some sense, the ``most classical'' clocks within the framework of~\cite{ClockModel}, in the sense that their internal state is assumed to decohere quickly (i.e., much faster than the time scale on which the clock operates).  As we will see, stochastic clocks exhibit some simple properties which make them particularly convenient to investigate. Specifically, they can be modelled as a  system with perfectly distinguishable and, in this sense, classical states. Furthermore, the sequence of states that a stochastic clock admits can be represented in terms of a Markov chain. 

Our main results  are  lower and upper bounds on the accuracy of stochastic clocks, measured in terms of the expected score~$N$ in the AT~Game that is achieved by the best strategy (Theorem~\ref{thm:GeneralClocksBound}).  The bounds depend on the size~$d$ of the clockwork. For stochastic clocks, $d$ corresponds to the number of classical states that the clockwork can admit. The bounds may be evaluated for specific clock constructions, such as \emph{stochastic ladder clocks}. We find that, for this type of clocks, the accuracy~$N$ grows linearly in~$d$ (Theorem~\ref{thm:TicksBound}).

Numerical results~\cite{ClockModel,ChrisThesis} suggest that general quantum clocks can achieve accuracies~$N$ that exceed this linear bound. That quantum coherence can indeed be exploited to build more accurate clocks has recently also been shown analytically for very large dimensions~\cite{RPaper}, although with respect to a different measure of accuracy (see Section~\ref{sec_discussion}  for more details).  It is therefore natural to conjecture that genuinely quantum clocks can outperform the best stochastic ones in the AT Game. The results presented here may be regarded as a stepping stone towards proving this conjecture.

    \section{The General Clock Model} \label{sec:classical-limit-model}
  
In order to facilitate comparison with other clock constructions, we will describe \emph{stochastic clocks}, which are the main subject of this paper, within the general framework of \emph{quantum clocks} introduced in~\cite{ClockModel}. The  idea underlying this general model is that  a clock can be regarded as a physical system equipped with a mechanism that emits time information. To describe this mechanism, one considers two systems:

\begin{itemize}
\item \emph{Clockwork (denoted $C$):} This is the running (dynamical) part of the clock. It interacts with its environment, thereby outputting {\it time information}.  

\item \emph{Tick registers (denoted $T_1, T_2, \ldots $):}  These belong to the environment of the clockwork. Each tick register~$T_i$ is, in sequential order, brought in brief contact with the clockwork to record time information. After this interaction, it is separated from the clockwork, keeping a record of the {\it clock time}. 
\end{itemize}

This idea is captured by the following definition.

 \begin{mydef}\label{def:QC}
  A \emph{quantum clock}  is defined by a pair $(\rho^0_C, \cM_{C \to CT})$, where $\rho^0_C$ is a density operator of a system $C$, called \emph{clockwork}, and $\cM_{C \to C T}$ is a  completely positive trace preserving (CPTP) map from $C$ to a composite system $C \otimes T$, where $T$ is called \emph{tick register}. 
\end{mydef}

The map $\cM_{C \to CT}$ may be decomposed into a part corresponding to an internal evolution $\cM^{\mathrm{int}}_{C \to C}$ of the clockwork and a part $\cM_{C \to CT}^{\mathrm{meas}}$ that extracts time information from the clockwork to the tick registers. However, one can also always take the internal evolution to be a trivial map, i.e. $\cM^{\mathrm{int}}_{C \to C}=\mathcal{I}$, and incorporate the entire clock dynamics into the action of the measuring map $\cM_{C \to CT}^{\mathrm{meas}}$. In the case of the \emph{stochastic clock} (Definition~\ref{def:StochasticMap} below), the clockwork consists of a finite number of ``classical'' states, and the evolution can be represented probabilistically, hence the label \textit{stochastic clock}. Inspired by the Salecker-Wigner-Peres clock model~\cite{Salecker1958, Peres}, we can picture a stochastic clock as a pointer moving on a circle, but with discretised dynamics. In other words, the clock's evolution can be characterised by the probabilities that determine how likely it moves from one state to another.  We will refer to these discrete movements as \emph{jumps}. 

Before proceeding to the formal definition of stochastic clocks, we first introduce a property for quantum clocks which captures the idea that the dynamics of any physically realistic clock must be continuous to a certain degree. The definition is also taken from~\cite{ClockModel}. 

        \begin{mydef}
    \label{def:ContClock}
    A  quantum clock $(\rho^0_C, \cM_{C \to C T})$ is called  \emph{$\eps$-continuous} if the map $\cM_{C \to C T}$ restricted to $C$ is $\eps$-close to the identity map $\cI_C$, i.e., 
    \begin{align}
      \frac{1}{2}\| \tr_{T} \circ \cM_{C \to C T} - \cI_C \|_{\diamond} \leq \eps
    \end{align}
    where $\| \cdot \|_{\diamond}$ is the diamond norm. 
     \end{mydef}
                   

 To fulfil this continuity condition for the clock, the total change of the clock state at every step has to be small enough. This means, in particular, that in any single application of the clock map, the occurrence of a tick must have a small probability.  We will show below that a stochastic clock whose jump probabilities are upper bounded by   $\eps$ is at least $\eps$-continuous.

     \section{Stochastic clocks} \label{subsec:classical-limit-model-mathematical-formulation}
     
      To define stochastic clocks we need to specify under what conditions such a clock produces a tick. We require that the tick registers output \textit{ticks} (represented by $\ket{1}_{T}$ in the two-dimensional tick register space) with a certain probability, whenever the clock system makes a jump that passes to or over the state $\ket{0}_{\cC}$ from a non-zero state. Note that without loss of generality one can choose any other state or states on the circle $\lbrace \ket{0},...,\ket{d} \rbrace$ for the tick states. For all other jumps the tick register would provide \textit{no-tick} --- events (represented by $\ket{0}_\mathrm{T}$). Note that one can think of these clocks as the {\it reset clocks}, since we will require that the clock starts over from the state $\proj{0}_C$ if the tick occurs.
      
In contrast to a general quantum clock, stochastic clocks are easier to analyse. Indeed, even for moderately large dimensions~$d$ of the clockwork~$C$, a general dynamics can lead to internal states of the clock that do not admit a simple description. 
         
  To understand the formal definition below, it is convenient to think of random variables $S_n$, for  $n \in \mathbb{N}_{\geq 0}$, which  represent the position on the $(d+1)$-dimensional clockwork  after $n$ applications of the clock map $\cM$.  Imagine that the clock state can jump to any of the states of the clockwork, which we denote by $\lbrace\ket{0},...,\ket{d} \rbrace$, with a given probability.
 We can represent this using Markov chains where $$P(S_{n+1}=l \mid S_n=k)=p_{k,l}\,  \quad n \in \mathbb{N}_{\geq 0}, \, k,l\in \lbrace 0,...,d \rbrace.$$
  Hence $P(S_{n+1}=l)=\sum\limits_{k=0}^d p_{k,l}P(S_n=k)$. This is  a  homogenous Markov chain, where the probability distribution does not depend on $n$. 
We will now define mathematically a general stochastic reset clock.
         \begin{mydef}
         \label{def:StochasticMap}
         {\it A stochastic reset clock} is given by a pair $(\rho_C^0,\cM_{C\to CT}^{\mathrm{st.}})$, as in Def.~\ref{def:QC}, with $\rho_C^0=\sum_j q_j\proj{j},\,j\in\lbrace 0,...,d\rbrace$ and the map $\cM^{\mathrm{st.}}:C \to CT$ taking the form:
         \begin{align*}
         \nonumber \cM^{\mathrm{st.}}(\rho_C)=\sum_{m=j}^{d}\sum_{j=0}^d \langle j|\rho_C|j \rangle p_{jm} \proj{m}_C \otimes \proj{0}_T + \\
         \sum_{m=0}^{j-1}\sum_{j=1}^d \langle j|\rho_C|j \rangle p_{jm} \proj{0}_C \otimes \proj{1}_T  \ .
         \end{align*}
          \end{mydef}
         
         Note that any passage over state $\proj{0}$ from a non-zero state is considered a tick, and hence after each detected tick the clock starts anew from the state $\proj{0}_C$. Without loss of generality one can take $\rho_C^0=\proj{0}$, and also choose any other state instead of $\ket{0}$ as a tick state of the clock. Values $p_{jk}$ represent a probability that clock jumps from the state $j$ to $k$, and as mentioned above the probability distribution for a jump is the same from any state, {\it i.e.} $p_{jk}=p_{0,(k-j) \mod d}$. Importantly, $p_{jk}\in [0,1],\,\forall j,k$ and $\sum\limits_{k=0}^d p_{jk}=1,\,\forall j \in \lbrace 0,...,d \rbrace$.\\
         
          We can show that the above defined $\cM_{C \to CT}^{\mathrm{st.}}$ is a valid clock map as prescribed by Def.~\ref{def:QC}.
       \begin{mylem}
       \label{Lemma1}
       The map $\cM_{C \to CT}^{\mathrm{st.}}$ from Def.~\ref{def:StochasticMap} is a completely positive trace preserving map.
       \end{mylem}
       \begin{proof}
       See Appendix.
       \end{proof}

        One can think of the simplest stochastic clock model, which is described by a toss of a coin at every step. We will call this the {\it stochastic ladder clock}. This clock was first studied in~\cite{ChrisThesis} (notice the usage of a similar clock in~\cite{ThermodynamicQClock} for thermodynamical considerations, and in~\cite{RPaper}, which is further discussed in the discussion section).
        
\begin{mydef}
 {\it A stochastic ladder clock} is a clock defined by Def.~\ref{def:StochasticMap} that can jump only to the next state with the jump probability $\delta\in [0,1]$ and stay in the current state with the probability $1-\delta$. Hence it is defined by:
\begin{align}
\label{eq:ladderClock}
\nonumber P(S_{n+1}=k+1\mid S_n=k)=p_{k,k+1}=\delta, \\
 P(S_{n+1}=k \mid S_n=k)=p_{k,k}=1-\delta
\end{align}
 independent of $n$ and $k$. 
        \end{mydef}
        
    We will next show that a stochastic clock can be made $\eps$-continuous, for arbitrary $\eps$, by choosing its probability distribution such that the total probability for a clock to move from its current state is sufficiently small.
      
  To obtain the diamond norm needed for continuity calculation, we will use the mathematical framework of semidefinite programming, as applied to the channel distance in~\cite{JohnWatrous}.
      
 For a quantum channel $\Phi\;:\;\sysC\rightarrow\sysC$, it holds that $\frac{1}{2}\|\Phi\|_\diamond\leq\eps$ if and only if
      \begin{equation}
      \label{eq:semidef}
            \begin{array}[t]{ll}
              \mbox{\textbf{there exists:}} & Z\geq J(\Phi),\\[1em]
              &Z\in \Pos\left(\sysC\otimes\sysC'\right)\\[1em]
              \mbox{\textbf{such that:}} &\|\tr_{\sysC}(Z)\|_\infty\leq\eps\\[1em]
            \end{array}        
      \end{equation} 
    where we take  $\sysC'$ isomorphic to $\sysC$ and  $J(\Phi)$ is a so-called Choi-Jamio\l{}kowski state,
    $$J(\Phi)=\sum\limits_{i,j}\Phi(\ket{i}\bra{j})\otimes \ket{i}\bra{j}_{C'}.$$
     In our case $\Phi = \tr_\mathrm{T}\circ\map_{\CtoCT}^{\mathrm{st.}}-\mathcal{I}_\mathrm{C}$.
    
     For our purpose it is enough to find a good upper bound for the channel distance, as the continuity holds with this upper bound as $\eps$.
      We will use the fact that any dual feasible operator $Z$ can be used to compute an upper bound on the optimal value of the dual problem~\ref{eq:semidef}. 
       
         Note that the values  $p_{0,k}$ will determine critically the possible values of $\eps$ for $\eps$-continuity. We will now show how exactly $\eps$ relates to the probability distribution of the jump of the clock.
         \begin{mylem}
         \label{Lemma2}
         The stochastic clock $(\rho_C^0,\cM_{C \to CT}^{\mathrm{st.}})$, with $\rho_C^0=\sum_j q_j\proj{j},\,j\in\lbrace 0,...,d\rbrace$ and $\cM^{\mathrm{st.}}$ of Def.~\ref{def:StochasticMap}, is $\eps$-continuous iff $\max\limits_{i \in \lbrace 0,...,d\rbrace}(1-p_{ii})\leq\eps$.
        \end{mylem} 
        \begin{proof}
        See Appendix. 
   \end{proof}
   
   Note that the statement of the lemma is rather intuitive. Indeed, a clock that moves from its current state with a large probability has a dynamics that does not approximate well the identity map.

\section{Performance of the stochastic clocks in the Alternate Ticks Game}
\label{sec:ATGame}

 As defined in~\cite{ClockModel}, the {\it Alternate Ticks (AT) Game} involves two collaborating players, $A$ and $B$, each equipped with a quantum clock as defined in Def.~\ref{def:QC}. The players can agree on a common strategy beforehand, but are not allowed to communicate once the game begins, nor to correlate their clocks prior to the start of the game. They are asked to provide ticks with their clocks in an alternating order --- {\it e.g.} first from $A$, then from $B$, then again from $A$, and so on. The goal of the players is to maximise the joint number of  ticks respecting the posed alternate ticks condition. 

We will now restrict to the case when clock starts from one of the basis states $\proj{j},\,j \in \lbrace 0,...,d \rbrace$. We describe the number of jumps (applications of the map $\cM$) needed for a clock to produce a tick $i$ after the tick $i-1$ by a random variable $Y_i$. Note that all $Y_i$ are independent, identically distributed ({\it iid}), since we consider reset clocks --- after the tick they always start anew from the state $\proj{0}$ and move with the same probability distribution. We will see that in the circular board picture, where the pieces move on the circle in one direction, the AT Game can be lost only if one player's peace overtakes the other for at least $d$ positions. This is however not sufficient as the other player can catch up with this fallout. We comment on this in more details below.

  We will assume that clocks can only move forward, but note that moving forward for $d$ steps is equivalent to moving backward for one step, since we can describe the dynamics on a circle. Hence we are still dealing with the most general case.
    To be able to analyse this further using probability theory, we will next define classical description of the stochastic clocks and the game dynamics. 
   
     \begin{mydef}
      Let $({\rho}_{\sysC}^{0},\cM_{\CtoCT}^{\mathrm{st.}})$ be a $d$-dimensional stochastic clock, with $\rho_C^0=\proj{j}$, for some $j \in \lbrace 0,...,d \rbrace $. We call the triple $\left(d,P^0,\{P^n\}_{n\in\mathbb{N}}\right)$ its classical description where $P^0 = j$ for $\rho_{\sysC}^{0} = \ket{j}\bra{j}$, $j\in\{0,\dots,d\}$  and the random variable $P^n$ denotes the number of states for which the clock has moved forward (has jumped over) after $n$ applications of the map $\map_{\CtoCT}^{\mathrm{st.}}$.
    \end{mydef}
    Note that the $k^{\mathrm{th}}$ tick happens after $\min \lbrace n:P^n=d k-1\rbrace$ steps. Now to quantify the success in the AT Game we need to look at the random variable describing the distance on the circular board between the players A and B. Suppose $P_\mathrm{A}^0\geq P_\mathrm{B}^0$, hence the player A is in front of the player B on the board.
       \begin{mydef}
    The \emph{classical game description} after $n$ applications of $\map_{\CtoCT}^{\mathrm{st.}}$, $\left(d,Q^0,\{Q^n\}_{n\in\mathbb{N}}\right)$, is defined by the random variable
    \begin{equation} \label{def:relative-position}
    Q^n := Q^0 + \left(P_\mathrm{A}^n - P_\mathrm{B}^n\right),
    \end{equation}
    where $Q^0 = P_\mathrm{A}^0 - P_\mathrm{B}^0$. We also introduce the random variables  $\Delta^n=Q^{n+1}-Q^{n}=P^{n+1}_{A}-P^{n}_{A}-(P^{n+1}_{B}-P^{n}_{B}),\,n\in\lbrace 0,1,...\rbrace$.
    \end{mydef}
    
       We now derive boundary cases for losing the game.
       
            \begin{mylem}
            \label{lem:NecessarySufficient}
     For the AT Game played with the stochastic clocks of dimension $d$ to be halted, it is necessary that one player is at least $d$ steps ahead of the other, {\it i.e.} $Q^n>d$ or $Q^n<0$. A sufficient criterion for halting is that one of the players is $2d$ steps in front of the other, {\it i.e.} $Q^n>2d$ or $Q^n<-d$.  
     \end{mylem}
     \begin{proof}
     See Appendix.
     \end{proof}
      The expected number of ticks in the AT Game is hence lower bounded by the number of ticks until one player is exactly $d$ steps in front of the other, and upper bounded by the expected number of ticks until one is exactly $2d$ steps in front of the other.
   
         Assume that each clock can move up to $m \leq d$ states in one jump. Then the distance between the positions of the clocks $A$ and $B$ can change for any value between $\lbrace 0,..., m \rbrace$ in one jump, \textit{i.e.} $\Delta^n=Q^{n+1}-Q^n \in \lbrace 0,..., m \rbrace$. Note that $p_{0,m}\neq 0$. We are interested in the minimal number of steps $n$ before the random variable $Q^n$ reaches one of the boundaries (one player being one or two circles in front of the other). This will lead to a higher-order differential equation, in the general case of order $d$. We will assume player $A$ was in front of the player $B$ at the start of the game, and $Q^0=z$.
         
         Denote with $T$ a random variable counting the number of ticks, $Y_i$ a random variable counting number of jumps between two ticks, and with $\mathbb{E}(S)=D_z$ expected number of jumps until the game is lost for initial distance between the clock states $Q^0=z$. Note that $Y=Y_i$ and $T$ are independent random variables. Then the following holds:
         \begin{align}
         \nonumber&\mathbb{E}(S)=D_z=\sum\limits_{i=1}^T Y_i=\mathbb{E}(Y\cdot T)=\mathbb{E}(Y)\cdot \mathbb{E}(T)=\mathbb{E}(Y)\cdot N\\
         \nonumber & N=\frac{D_z}{\mathbb{E}(Y)}
         \end{align}
          \\
            where $N$ is the expected number of alternate ticks.  In general, for the expected number of jumps $S$ until the end of the game (number of applications of the map $\cM$), we have:
               \begin{equation}
              \label{eq:recurrence}
              \begin{split}
             &\mathbb{E}(S|Q^0=z)=D_z= p_{1}(D_{z+1}+1)+p_2(D_{z+2}+1)\\
             &+...+p_m(D_{z+m}+1)+ p_0(D_z+1)+p_{-1}(D_{z-1}+1)\\
             &+p_{-2}(D_{z-2}+1)+...+p_{-m}(D_{z-m}+1)
             \end{split}
             \end{equation}
              where $m<d$ is the maximal number of states for which the positions of A and B can differ in one jump, and $p_k$ denotes the probability that the difference between positions of players $A$ and $B$ on the circle changed for exactly $k$ in one jump. 
              We can have maximum $m=d$ which is the case that one clock jumps for $d$ steps and the other stays where it was.\\
              Note that the value of $m$ will be important in determining $\eps$-continuity of a clock, since the bigger possible $m$ we have, the less possibilities we have for small $\eps$. In other words, the probability for the clock to move for many states in the same jump has to remain low for an $\eps$-continuous clock with small $\eps$.
              
              \subsection{Identical clocks in the Alternate Ticks Game} 
              
             If the two clocks are identical, the probability to change the distance between A and B for $+k$ or $-k$ states in one step will be the same. This means that variable $Q^n$ changes for $+k$ or for $-k$ with the same probability, leading to $p_k=p_{-k}\,, \forall k \in \lbrace 1,...,m \rbrace$.  Note also that $$\sum\limits_{k=-m}^m p_k=2\sum\limits_{k=1}^m p_k+p_0=1$$ 
             
             Hence our recurrence relation becomes:
                \begin{equation}
              \label{eq:recurrenceGen}
              \begin{split}
              (2\sum\limits_{k=1}^m p_k)D_z = \\
              p_{1}D_{z+1}+p_2D_{z+2}+...+p_mD_{z+m}+ \\
              p_{1}D_{z-1}+p_{2}D_{z-2}+...+p_{m}D_{z-m}+1
             \end{split}
             \end{equation}
             
             Solving this relation (done in details in Appendix) and inputing optimal $z=\frac{d}{2}$, we obtain:
             \begin{equation}
             \label{eq:boundGeneral}
            \frac{d^2/4+d+1}{2\sum\limits_{k=1}^m k^2p_k} \leq D_{\frac{d}{2}} \leq \frac{9d^2/4+3d+1}{2\sum\limits_{k=1}^m k^2p_k}
            \end{equation}
             and we know that the expected number of ticks $N=\mathbb{E}(T)=\frac{D_z}{\mathbb{E}(Y)}$. Note that for for $\eps=0$ we have a clock described completely by the identity map, that will not produce a tick. Hence $\mathbb{E}(Y)=\infty$ for this clock and it can produce zero expected alternate ticks. Note also that by the definition of the diamond norm $\eps \leq 1$. Clock with $\eps=1$ will be described in the next section.
             
We can then show the following Theorem.
             
             \begin{mythm} 
             \label{thm:GeneralClocksBound}
             Suppose that the AT Game is played between two identical $\eps$-continuous $d$-dimensional stochastic clocks, with $\eps\in(0,1)$. Then the expected number of alternate ticks $N$ is bounded by
             $$\frac{(d+4+\frac{4}{d})p_{i,i+1}}{4m^2\eps(2-\eps)} \leq N \leq \frac{\frac{9d}{4}+3+\frac{1}{d}}{2(1-\eps)mp_{0,m}}$$
             where $m$ is the furthest state the clock can jump to from the state $0$ with probability $p_{0,m}\neq 0$, and $p_{i,i+1}$ is the probability for the clock to jump from state $i$ to $i+1$.
             \end{mythm}
             
             \begin{proof}
             See Appendix.
             \end{proof}

We can see from the above result that for $mp_{0,m}$ small, we have a greater linear upper bound on the number of ticks, and for $mp_{0,m}$ approaching $d$ we have a very bad constant bound. This suggest that the ladder clock is one of the best stochastic clocks one can build, since $mp_{0,m}=\delta \leq \eps$.\\
   To maximise lower bound one needs $m^2\eps$ to be small, which again is minimised for a ladder clock (where $m=1$) with small $\eps$. \\

\subsection{Perfect stochastic clocks}

We will now introduce a clock that we will call {\it the perfect clock}. When the  AT Game is played with two perfect clocks, it always produces infinite number of ticks.  However, the perfect clock is not $\eps$-continuous for $\eps<1$, i.e., it is not a physical clock (see~\cite{ClockModel}).

\begin{mydef} A {\it perfect stochastic clock} is a clock that always produces a tick after exactly the same number of applications of the map $\cM$.
\end{mydef}

Again note that this clock is not $\eps$-continuous for $\eps < 1$ since $1-p_{ii}=1$, and that if two identical perfect clocks play the AT Game then $p_k =0,\,k \geq 1$. Hence from Eq.~\ref{eq:boundGeneral} the expected number of ticks is $N=\mathbb{E}(T)=\infty$.

In the ladder clock model, this clock will always produce a tick after $D-1$ steps, where $D$ is the dimension of the clockwork system. Since we have taken that $D=d+1$, the tick is produced after exactly $d$ steps. Hence $\mathbb{E}(Y)=d$ for the perfect stochastic ladder clock. 

In the Appendix we show that the following bunds hold if this clock was to play the AT Game with any stochastic ladder clock:
  \begin{equation}
  \frac{d+1}{1-\delta}\leq D_0\leq\frac{2d+1}{1-\delta}
  \end{equation}
  Then $N=D_0/d$ gives a bound on the number of alternate ticks, which will be very bad for small $\delta$, as expected.
  If the perfect ladder clock plays against the perfect clock we have $\delta=1$ and they can play the AT Game forever.

        \section{Bounds on the expected number of alternate ticks  for  ladder clocks} \label{sec:result}

 For any clockwork dimension $d$, we will derive a tight lower and upper bound on the expected number of alternate ticks generated in the AT Game with two identical ladder clocks. 
    
     The lower bound obtained also provides us with the minimum average number of alternate ticks achievable with the best possible quantum clocks within the given mathematical framework. It is an open question whether it is possible to achieve a stronger-than-linear relationship, but current numerical results - {\it e.g.} \cite{ClockModel, ChrisThesis} and work on different clock accuracy definitions~\cite{RPaper} confirm such a superiority of well-defined quantum to any classical clocks for certain measures used.

     Remember from Eq.~\ref{eq:ladderClock}, that for the ladder clock $p_{ii}=1-\delta,\, \forall i \in \lbrace 0,...,d \rbrace$ and $p_{i,i+1}=\delta$. One can easily deduce that the possible change in distance between two players in one step takes values: $\Delta^n \in \lbrace 0,1,-1\rbrace, \forall n$, with $P(\Delta^n=0)=\delta^2+(1-\delta)^2=r$ when either both A and B stand still or both make a step, $P(\Delta^n=1)=\delta(1-\delta)=p$ when A makes a step and B stands still and $P(\Delta^n=1)=\delta(1-\delta)=q$ when B makes a step and A stands still. \\
         The connection to the Gambler's ruin problem can be seen easily~\cite{ProbabilityTheory}. Consider scenario where at each round gambler wins with probability $p$, loses with probability $q$ or has a draw with a casino with the probability $r=1-p-q$. \\
    The following result can easily be deduced from our calculations in Sec.~\ref{sec:ATGame}. 

  \begin{mythm} \label{thm:TicksBound}
    Consider the AT Game played with two identical $d$-dimensional stochastic ladder clocks with a step probability $\delta$. Assume that an optimal strategy is applied. Then the expected number of alternate ticks 
     produced on each side is lower and upper bounded by functions asymptotically linear in the dimension $d$ of the clockwork space, given by
        \begin{equation} \label{eq:Bounds}
          \begin{aligned}
            &N_\mathrm{lower}(d) = \frac{1}{8(1-\delta)}d + \frac{1}{2(1-\delta)}+O\Big(\frac{1}{d}\Big)\\[1em]
            &N_\mathrm{upper}(d) = \frac{5}{8(1-\delta)}d + \frac{3}{2(1-\delta)}+O\Big(\frac{1}{d}\Big)
          \end{aligned}
        \end{equation}
      respectively.
    \end{mythm}
    
    \begin{proof}
    Note that the expected number of steps in the random walk until one of the boundaries is reached, is also the number of times $S$ the map $\map_{\CtoCT}^{\mathrm{st.}}$ will be applied on both sides A and B. 
    For the ladder clock we have $Y=\sum\limits_{i=0}^{i=d} X_i$ where $X_i$ are {\it iid} random variables denoting the number of applications of the map to go from position $i$ to $i+1$. \\
    
    Note that $\mathbb{E}(X_i)=\frac{1}{\delta}$ and hence $\mathbb{E}(Y)=\frac{d}{\delta}$, leading to $$N=\mathbb{E}(T)=\mathbb{E}(S) / \mathbb{E}(Y)=\frac{\delta}{d}\mathbb{E}(S)$$
     Now we are left with calculating $\mathbb{E}(S)$. But this is also the expected number of steps until random walk is ended for variable $Q^n$. Again call $\mathbb{E}(S)=D_z$ for game starting in $Q^z$. Then we obtain the following recurrence relation:
    $$D_z=pD_{z+1}+rD_z+qD_{z-1}+1$$
    with $p+q+r=1, p=q=\delta(1-\delta)$.\\
     Absorbing boundaries are $D_{-1}=D_{d+1}=0$ for the lower bound, and $D_{-d-1}=D_{2d+1}=0$ for the upper bound on the game (see Lemma~\ref{lem:NecessarySufficient}).
    The general solution of this recurrence relation is $$D_z=A+Bz-\frac{z^2}{2\delta(1-\delta)}$$
    Inserting boundary conditions, and maximising $D_z$ over $z$, to obtain the optimal starting conditions of the game (the corresponding z is $d/2$ which is as expected - it is the best strategy to start with diametrically opposite states of the two identical clocks), we obtain bounds from Eq.~\ref{eq:Bounds}.
    \end{proof}
 
     \begin{figure}[h!]
        \includegraphics[scale=0.3]{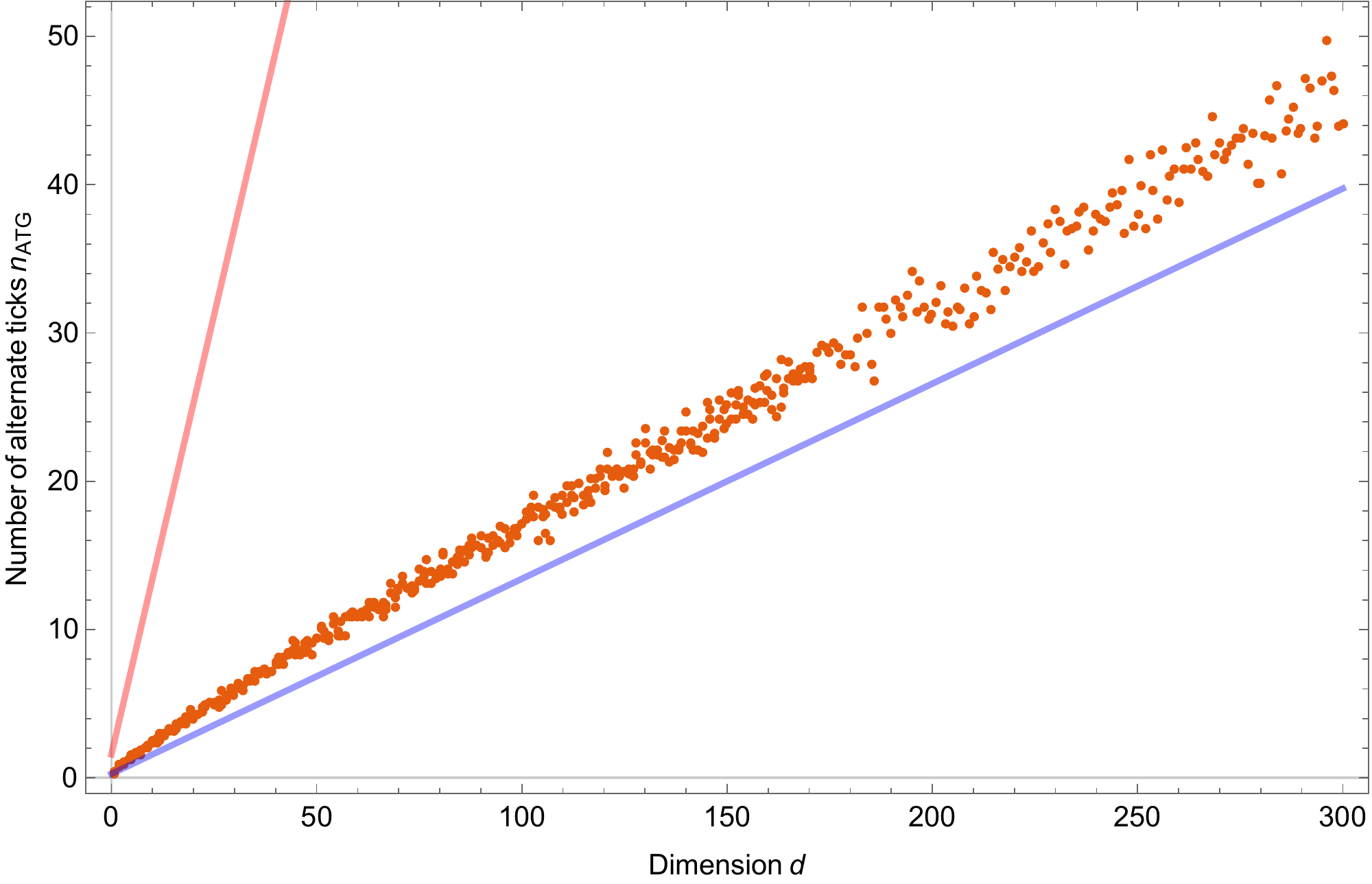}
        \caption{\textbf{Expected number of the alternate ticks of the two identical stochastic ladder clocks.} The blue line represents an asymptotic lower bound and the red line an asymptotic upper bound. The orange data points are results from numerical simulations with the step parameter $\delta=0.05$, averaged over $500$ iterations.}
         \label{ATbound}
    \end{figure}   
         
         In Fig.~\ref{ATbound} we show numerical results for the performance of the stochastic ladder clock in the AT Game. Note that the numerics confirms our derived analytical bounds.

\section{Discussion and Conclusions} \label{sec_discussion}

This work is motivated by the principle that time, rather than just being an external parameter in the  equations of motion, should be defined operationally as the quantity we can read off a clock. The principle entails that  measures for the accuracy of clocks must not depend on an external time parameter, either. Indeed, having access to such a parameter would  be akin to using a perfectly accurate clock as a reference. Conversely, the score~$N$ achievable in the AT~Game, which we considered here, is defined without any reference to parametric time.  Note that this operational viewpoint is also imposed on us when we talk about the accuracy of the best available clocks, e.g.,  atomic clocks~\cite{AtomicClocks,Wineland,WinelandFrequencyStandards}. Since there do not exist more accurate clocks, we have no reference to test against  --- but we can of course still count  how many ticks two such  clocks  emit until they run out of synchronisation.

The performance of a clock in the AT~Game is sometimes difficult to determine analytically. One may thus try to approximate it by measures with more mathematical structure. One such possible measure is the so-called \emph{$R$-value}, defined as $R=\smash{\frac{\mu^2}{\sigma^2}}$, where $\mu$ is the expected length of the time interval between two ticks and $\sigma$ its standard deviation. The drawback of this measure is that it uses the parametric time explicitly. Intuitively, one may think of it as a measure for how many ticks it takes until the standard deviation in  the (parametric) time of a tick is equally large as the average time interval between two ticks. This interpretation  suggests that the $R$-value of a clock can related to the AT~Game score~$N$. In work simultaneous to this one~\cite{RPaper}, further evidence for the close interrelation between $R$ and $N$ is provided. In particular, it has been shown that $R$ is upper bounded by a function linear in the size~$d$ for stochastic clocks and that it equals $d$ for the first tick of stochastic ladder clocks.  

The precise relation between the AT~Game performance~$N$ and the $R$-value for general clocks is however unknown. Given the above considerations, a natural conjecture would be that they depend linearly on each other. This would have interesting consequences. In~\cite{RPaper}, a construction of a quantum clock has been presented whose $R$-value grows almost quadratically in its size~$d$, when $d\rightarrow \infty$. If the conjecture was true, this would imply that quantum clocks can outperform the best stochastic clocks in the AT~Game. 
             
We also note that the operational accuracy measure based on the AT~Game that we employed here may be generalised in various directions. For example, one could consider a variant of the AT~Game with more than two players and strengthen the condition of alternating ticks accordingly. Another possibility is to consider non-binary tick registers and demand that the emitted ticks all contain a different number.  It would also be interesting to investigate scenarios that include relativistic effects, like the ones described in {\it e.g.}~\cite{CaslavEsteban,LockFuentes}. Another, possibly quite far-reaching, project would be to apply the operational approach taken here to the correlations framework of~\cite{PageWootters,Wootters,Mott,Briggs,Brout,GiovannettiQuantumTime,ChiaraVlatko}. In particular, the clock systems employed within the latter may be replaced by finite-size clocks as considered here (see~\cite{PanticMasterThesis} for some initial steps in this direction).

   \section*{Acknowledgements}

The authors thank Ralph Silva for insightful discussions. We acknowledge support from the Swiss National Science Foundation (SNSF) via the National Centre of Competence in Research QSIT and project No.\ $200020_{-}165843$.  

\appendix  
\section{Proof of Lemma~\ref{Lemma1}}

   The map $\cM^{\mathrm{st.}}$ of Def.~\ref{def:StochasticMap} is a completely positive trace preserving map.
       \begin{proof}
       By Choi's theorem~\cite{ChoiCompletelyPositive}, it is enough to show that there is a Kraus representation for the map $\cM$.
       Define: 
       \begin{align}
        \nonumber&M_{im}=\sqrt{p_{im}}\ket{0}\bra{i}_C \otimes \ket{1}_T,\, i \neq 0, m \in \lbrace 0,...,i-1 \rbrace \\
        \nonumber&M_{im}= \sqrt{p_{im}}\ket{m}\bra{i}_C \otimes \ket{0}_T, \, m \in \lbrace i,...,d \rbrace
        \end{align}
        It can easily be seen that these operators are linear operators from $C$ to $CT$. Also 
        \begin{align}
        M_{im}^{\dagger}M_{im}=p_{im}\proj{i},\, \forall i,\, m
        \end{align}
         and
         \begin{align}
         \nonumber\sum\limits_{i=0}^d\sum\limits_{m=0}^d M_{im}^{\dagger}M_{im}=\sum\limits_{i=0}^d(\sum\limits_{m=0}^d p_{im})\proj{i}=\sum\limits_{i=0}^d\proj{i}=\id_C
         \end{align}
          Hence they represent Kraus operators and we can easily see that $$\cM^{\mathrm{st.}}(\rho_C)=\sum\limits_{i=0}^d\sum\limits_{m=0}^dM_{im}\rho_C M_{im}^{\dagger}$$ is indeed a CPTP map.\\
       \end{proof}

  \section{Proof of Lemma~\ref{Lemma2}}
        The stochastic clock $(\rho_C^0,\cM^{\mathrm{st.}})$, with $\rho_C^0=\sum_j q_j\proj{j},\,j\in\lbrace 0,...,d\rbrace$ and $\cM$ of Def.~\ref{def:StochasticMap}, is $\eps$-continuous iff $\max\limits_{i \in \lbrace 0,...,d\rbrace}(1-p_{ii})\leq\eps$.
      
        \begin{proof}
        First note that since we are considering that random variable denoting the position of a stochastic clock performs a homogenous Markov chain, we have that $p_{ii}$ is a constant for all $i \in\lbrace 0,...,d \rbrace$. We will now use semi-definite programming described above to obtain an upper bound on $\bigl\| \tr_T\cM^{\mathrm{st.}}-\id_C \bigr\|_{\diamond}= \bigl\| \Phi \bigr\|_{\diamond}$ which will also determine possible values of $\eps$ from Def.~\ref{def:ContClock}.\\
        As a first step we need to calculate Choi-Jamio\l{}kowski state of $\Phi=\tr_T\cM^{\mathrm{st.}}-\id_C:C\rightarrow C$ which is, assuming $C'$ to be a copy space of $C$, $J(\Phi)=\sum\limits_{i,j}\Phi(\ket{i}\bra{j})\otimes \ket{i}\bra{j}_{C'}$. Note that         
         \begin{equation}
          \Phi(\ket{i}\bra{j}) = \left\{\begin{array}{ll}
            -\ket{i}\bra{j} &\mbox{for } i\neq j\\[1em]
            -\proj{i}+\sum\limits_{m=i}^d p_{im} \proj{m} +\\[1em]
            \sum\limits_{m=0}^{i-1} p_{im} \proj{0}&\mbox{for } i=j \geq 1 \\[1em]
            -\proj{0}+\sum\limits_{m=0}^{d} p_{0m} \proj{m}&\mbox{for } i=j=0
          \end{array}\right.
      \end{equation}
hence we have:
\begin{align}
\nonumber&J(\Phi)= \\
\nonumber&-\sum\limits_{i=0}^d\ket{i}\bra{i}\otimes\proj{i}_{C'}+\sum\limits_{i=0}^d\sum\limits_{m=i}^d p_{im} \proj{m}\otimes \proj{i}_{C'}\\
&+\sum\limits_{i=1}^d\sum\limits_{m=0}^{i-1} p_{im} \proj{0}\otimes \proj{i}_{C'}-\sum\limits_{i\neq j}\ket{i}_C\bra{j}\otimes\ket{i}_{C'}\bra{j}
 \end{align}       
 Now we need to find a dual feasible map $Z$ such that both $Z$ and $X=Z-J(\phi)$ are positive maps on $\cC \otimes \cC'$.
   Let $$X=\sum\limits_{i\neq j}\ket{i}\bra{j}_C \otimes\ket{i}\bra{j}_{C'}+\sum\limits_{i=0}^d (1-p_{ii}) \proj{i}\otimes \proj{i}$$ from which it follows 
   \begin{align}
   \nonumber&Z=\sum\limits_{m=i+1}^d\sum\limits_{i=0}^d p_{im} \proj{m}_C\otimes \proj{i}_{C'}\\
   &+\sum\limits_{m=0}^{i-1}\sum\limits_{i=1}^d p_{im} \proj{0}_C\otimes \proj{i}_{C'}
   \end{align} 
   Now let $\ket{\phi}=\sum_{i,j}z_{i,j}\ket{i}_C\otimes\ket{j}_{C'} \in C\otimes C'$. Then to show $X$ and $Z$ are positive we need to show that $\langle \phi | X | \phi \rangle \geq 0$ for any $\phi$ and correspondingly for $Z$. Calculating,
        \begin{align*}
            &\bra{\phi}X\ket{\phi}= \\
            &\sum_{i,j,k,l,m,n}z_{ij}^*z_{mn}\left(\bra{i}\otimes\bra{j}\right)\\
            &\left( \ket{k}\bra{l}\otimes\ket{k}\bra{l}+ (1-p_{kk}) \proj{k}\otimes \proj{k}\right)\left(\ket{m}\otimes\ket{n}\right)\\
             &=\sum_{k,l}z_{kk}^*z_{ll}+ \sum\limits_{i=0}^d (1-p_{ii}) |z_{ii}|^2\\
             &=\sum_{k<l}\left(z_{kk}^*z_{ll} + z_{kk}z_{ll}^*\right) + \sum_k \left|z_{kk}\right|^2+ \sum\limits_{i=0}^d (1-p_{ii}) |z_{ii}|^2\\
            &=\sum_{k<l}2\mathrm{Re}\left(z^*_{kk}z_{ll}\right) + \sum_k \left(\mathrm{Re}^2(z_{kk}) + \mathrm{Im}^2(z_{kk})\right) \\
            &+\sum\limits_{i=0}^d (1-p_{ii}) |z_{ii}|^2
          \end{align*}
        
        Now, notice that $\mathrm{Re}(z_{kk}^*z_{ll}) = \mathrm{Re}(z_{kk})\mathrm{Re}(z_{ll}) + \mathrm{Im}(z_{kk})\mathrm{Im}(z_{ll})$ and the sums can be written as:
       
   $$\langle \phi | X | \phi \rangle = (\sum\limits_{i=0}^d \mathrm{Re}{z_{ii}})^2 + (\sum\limits_{i=0}^d \mathrm{Im}{z_{ii}})^2 + \sum\limits_{i=0}^d (1-p_{ii}) |z_{ii}|^2 \geq 0$$
   
   since all of the terms are non-negative. Also
        \begin{align}
   \nonumber\langle \phi | Z | \phi \rangle = \sum\limits_{i=0}^d\sum\limits_{m=i+1}^d  p_{im} |z_{mi}|^2  + \sum\limits_{i=1}^d\sum\limits_{m=0}^{i-1}  p_{im} |z_{0i}|^2 \geq 0
      \end{align}
   
   Further, we can see that
   \begin{align}
   \nonumber&\tr_C Z=\sum\limits_{i=0}^d\sum\limits_{m=0}^d p_{im} \proj{i}_{C'} - \sum_{i=0}^d p_{ii} \proj{i}_{C'} \\
   &=\id_{C'}- \sum_{i=0}^d p_{ii} \proj{i}_{C'}
   \end{align}
   Hence $\tr_C Z$ is a diagonal matrix with elements $1-p_{ii}$ on the main diagonal.
   We can hence deduce that $||\tr_C Z||_{\infty}=\max\limits_{i=0,...,d} (1- p_{ii})=1-p_{00}$ without loss of generality.
   Since $||\tr_C Z||_{\infty}$ represents a bound for the  $\frac{1}{2}\bigl\| \Phi \bigr\|_{\diamond}$we see that our stochastic clock is $\eps$-continuous for all $\eps \geq (1-p_{00})$.
   Also, given an $\eps$ we can choose a probability distribution for our clock jumps to be such that $1-p_{ii} \leq \eps$ and we know that our clock will be $\eps$-continuous by Def.~\ref{def:ContClock}. 
        \end{proof}

 \section{Proof of Lemma~\ref{lem:NecessarySufficient}}
  
        Let Alice and Bob play the AT Game with two identical $d$-dimensional stochastic clocks A and B and suppose the game has ended after $n_\mathrm{end}\geq 1$ applications of map $\cM_{C \to CT}$. Then either A or B has to be at least $d$ states ahead at moment $n_\mathrm{end}$.

    \begin{proof}
    Let $\left(d,P^0_\mathrm{A},\{P^n_\mathrm{A}\}_{t\in\mathbb{N}}\right)$ and $\left(d,P^0_\mathrm{B},\{P^n_\mathrm{B}\}_{t\in\mathbb{N}}\right)$ be the classical dynamical descriptions of the clocks and let $\left(d,Q^0,\{Q^n\}_{n\in\mathbb{N}}\right)$ be the classical game description. Assume that it is possible for the AT Game to end without any of the clocks being $1$ time ahead.
    
    First, consider two special cases: 1) The game ends before B ticks for the first time: In this case, $P_\mathrm{A}^{n_\mathrm{end}} = 2d-P_\mathrm{A}^0$ and $P_\mathrm{B}^{n_\mathrm{end}}\leq d-P_\mathrm{B}^0 - 1$, and therefore $Q^{n_\mathrm{end}} \geq d+1$, i.e. A is $d+1$ states ahead. 2) The game ends before A ticks for the first time: Here, $P_\mathrm{B}^{n_\mathrm{end}} = d-P_\mathrm{B}^0$ and $P_\mathrm{A}^{n_\mathrm{end}} \leq d - P_\mathrm{A}^0 - 1$, and therefore $Q^{n_\mathrm{end}}\leq -1$, i.e. B is $d+1$ states ahead.
    
    Next, assume that each clock has ticked at least once. Let $\bar{n}_\mathrm{A}, \bar{n}_\mathrm{B} < n_\mathrm{end}$ be the times at which clocks A and B produce their last \textit{alternate} tick, respectively. Now, if A produces the final (\textit{non-alternate}) tick that ends the game at moment $n_\mathrm{end}$, then $P^{n_\mathrm{end}}_\mathrm{A} = P^{\bar{n}_\mathrm{A}}_\mathrm{A} + d$ and $P^{n_\mathrm{end}}_\mathrm{B}\leq P^{\bar{n}_\mathrm{B}}_\mathrm{B} + d-1$. Remind yourself that a clock generates its $k$th tick at the moment corresponding to the smallest $n$ for which $P^n = kd-P^0$. Since A produced the first tick by convention and the last tick by assumption, we know that $P_\mathrm{A}^{\bar{n}_\mathrm{A}} = kd-P_\mathrm{A}^0$ and $P_\mathrm{B}^{\bar{n}_\mathrm{B}} = (k-1)d-P_\mathrm{B}^0$ for some $k\geq2$ alternate ticks. Solving one of the equations for $k$ and plugging it into the other one yields: $P^{\bar{n}_\mathrm{A}}_\mathrm{A} = P^{\bar{t}_\mathrm{B}}_\mathrm{B} + (d-Q^0)$ and thus $Q^{n_\mathrm{end}}=P^{n_\mathrm{end}}_\mathrm{A}-P^{n_\mathrm{end}}_\mathrm{B}+Q^0\geq d+1$. But by definition, this means that A is $d+1$ states ahead.
    
    Similarly, if B produces the final tick that ends the game, then $P^{n_\mathrm{end}}_\mathrm{B} = P^{\bar{n}_\mathrm{B}}_\mathrm{B} + d$ and $P^{n_\mathrm{end}}_\mathrm{A}\leq P^{\bar{n}_\mathrm{A}}_\mathrm{A} + d-1$. But following an analogous argument as above, we find $P^{\bar{n}_\mathrm{B}}_\mathrm{B}=P^{\bar{n}_\mathrm{A}}_\mathrm{A}+Q^0$, so that $Q^{n_\mathrm{end}} = P^{n_\mathrm{end}}_\mathrm{A}-P^{n_\mathrm{end}}_\mathrm{B}+Q^0 \leq -1$ and B is at least $d+1$ states ahead.
    \end{proof}

    Let Alice and Bob play the AT Game with identical $d$-dimensional stochastic clocks A and B of dimension $d$ and suppose that one of the clocks is $2$ times ahead (at least $2d+1$ steps ahead) at some moment $n\in\mathbb{N}$. Then there exists $n_\mathrm{end}\leq n$, for which the Alternate Ticks Game has ended.

       \begin{proof}
    Let again $\left(d,P^0_\mathrm{A},\{P^n_\mathrm{A}\}_{n\in\mathbb{N}}\right)$ and $\left(d,P^0_\mathrm{B},\{P^n_\mathrm{B}\}_{n\in\mathbb{N}}\right)$ be the classical dynamical descriptions of the clocks and let $\left(d, Q^0,\{Q^n\}_{n\in\mathbb{N}}\right)$ be their classical game description. Notice that for a given classical dynamical description $\left(d,P^0,\{P^n\}_{n\in\mathbb{N}}\right)$ for a stochastic clock, for any moment $n$, we have seen that the number of alternate ticks is given by 
    $k_n= \left\lfloor\frac{1}{d}\left(P^n+P^0\right)\right\rfloor$.

    Now, assume first, that A is $2$ times ahead at moment $n$. Then $Q^n \geq 2d + 1$, i.e. $P_\mathrm{A}^n+P_\mathrm{A}^0\geq P_\mathrm{B}^n + P_\mathrm{B}^0 + 2d + 1$. From the above equation, it follows that A has produced at least two more ticks than B. But by the \textit{pigeon hole principle}, this immediately implies that at least two ticks of A must have occurred \textit{non-alternately}, i.e. without B ticking in between. Therefore, the game must have ended at some earlier moment $t_\mathrm{end}\leq n$.
    
    Similarly, notice that if B is $2$ times ahead at some moment $n$, it holds that $Q^n\leq -d-1$, so that $P_\mathrm{B}^n+P_\mathrm{B}^0 \geq P_\mathrm{A}^n+P_\mathrm{A}^0 +d +1$, i.e. B ticked more often than A. By definition, A must have produced the first tick (otherwise the game has already ended and we are done). But that again means that B must have ticked $2$ times more than A after the first alternate tick. By the \textit{pigeon hole principle} two ticks by B must have violated the alternating order and the game ended at some $n_\mathrm{end}\leq n$.
    \end{proof}
    
     \section{Proof of Theorem~\ref{thm:GeneralClocksBound}}
    
             Assume that the AT Game is played between two identical stochastic clocks. Then following bounds hold on the number of alternate ticks:
             \begin{equation}
             \frac{(d+4+\frac{4}{d})p_{i,i+1}}{4m^2\eps(2-\eps)} \leq N \leq \frac{\frac{9d}{4}+3+\frac{1}{d}}{2(1-\eps)mp_{0,m}}
             \end{equation}
       
      \begin{proof}       We can consider homogenous equation from Eq.~\ref{eq:recurrenceGen} (which excludes the constant term $1$ on the RHS). One can notice that $A+Bz$ will be the solution. Namely inserting $D_z^{\mathrm{hom}}=A+Bz$ to the homogenous part of Eq.~\ref{eq:recurrenceGen}, we obtain:
              \begin{equation}
              \label{eq:recurrenceHom}
              \begin{split}
              &(2\sum\limits_{k=1}^m p_k)(A+Bz) = \\
              & p_{1}(A+Bz+B)+...+p_m(A+Bz+mB)+ \\
             & p_{1}(A+Bz-B)+...+p_{m}(A+Bz-mB)\\
             &=(2\sum\limits_{k=1}^m p_k)A+(2\sum\limits_{k=1}^m p_k)Bz
             \end{split}
             \end{equation}
             Note that this is the case because of the symmetry of the particular problem.
             Hence we are left to find a particular solution. After considering a quadratic function in $z^2$, we obtain following solution:
             $$D_z^{\mathrm{part}}=-\frac{z^2}{2\sum\limits_{k=1}^m k^2p_k}$$
             Hence the total solution is $D_z=D_z^{\mathrm{hom}}+D_z^{\mathrm{part}}=A+Bz-\frac{z^2}{2\sum\limits_{k=1}^m k^2p_k}$. \\
             We can now impose boundary conditions (recall: $D_{-1}=D_{d+1}=0$ for the lower bound, and $D_{-d-1}=D_{2d+1}=0$ for the upper bound on the game). For the lower bound we have
             \begin{equation}
              A=\frac{d+1}{2\sum\limits_{k=1}^m k^2p_k},\, B=\frac{d}{2\sum\limits_{k=1}^m k^2p_k}
              \end{equation}
               This leads to
             \begin{equation}
             D_z^{\mathrm{lower}}=\frac{d+1+dz-z^2}{2\sum\limits_{k=1}^m k^2p_k}
             \end{equation}
             For the upper bound we obtain 
             \begin{equation}
             A=\frac{(d+1)(2d+1)}{2\sum\limits_{k=1}^m k^2p_k},\,B=\frac{d}{2\sum\limits_{k=1}^m k^2p_k}
             \end{equation}
             and
             \begin{equation}
             D_z^{\mathrm{upper}}=\frac{2d^2+3d+1+dz-z^2}{2\sum\limits_{k=1}^m k^2p_k}
             \end{equation}

             Now recall that $\mathbb{E}(S)=\mathbb{E}(Y\cdot T)=\mathbb{E}(Y)\cdot \mathbb{E}(T)$, where $T$ is the number of ticks and $\mathbb{E}(S)=D_z$ for initial distance between clock states $Q^0=z$. We will have $D_z$ maximised for $z=d/2$  as expected for identical clocks (they start from opposite states on the circle). Note that $\mathbb{E}(Y)$ is not trivial to calculate as it depends on the probability distribution of the particular clock. \\
             Inputing $z=\frac{d}{2}$, we have:
             \begin{equation}
            \frac{d^2/4+d+1}{2\sum\limits_{k=1}^m k^2p_k} \leq D_z \leq \frac{9d^2/4+3d+1}{2\sum\limits_{k=1}^m k^2p_k}
            \end{equation}

              Now write $\mathbb{E}(Y)=\sum yP(Y=y)$. We can restrict to the case when clock can not jump more than $\left\lfloor \frac{d}{s} \right\rfloor=m\geq 1$ places in one step. This is the same $m$ as the maximum number of place differences players can achieve in one step. Note also that the average number of steps each clock needs until the tick is reached is then lower bounded by $s$, $\mathbb{E}(Y)\geq s$. Hence, expected number of ticks can be upper bounded as follows (note that $p_m > 0$ by assumption):
      \begin{align}
      \nonumber&N=\mathbb{E}(T) \leq \frac{D^{upper}_{d/2}}{\mathbb{E}(Y)}\leq \frac{\frac{9d^2}{4}+3d+1}{2s\sum\limits_{k=1}^m k^2p_k}\\
      \nonumber& \leq \frac{\frac{9d^2}{4}+3d+1}{2s m^2p_m}=\frac{\frac{9d}{4}+3+\frac{1}{d}}{2mp_m}\\
      &\leq \frac{\frac{9d}{4}+3+\frac{1}{d}}{2(1-\eps)mp_{0,m}}
      \end{align}
   using $1-p_{ii}\leq \eps$, $p_m \geq p_{ii}^A p_{0,m}^B \geq (1-\eps) p_{0,m}$, $s m =d$ and $m\geq 1$. Note that $p_{ii}$ is a constant for any $i$ and same for both clocks. We can see from the above result that for $mp_{0,m}$ small, we have a better linear bound on the number of ticks, and for $mp_{0,m}$ approaching $d$ we have a very bad constant bound. This suggest that the `ladder' clock is one of the best stochastic clocks one can build.\\
   
    Note also that the average number of steps until the tick of one clock, $\mathbb{E}(Y)$ is upper bounded by $\frac{d}{p_{i,i+1}}$, where $p_{i,i+1}$ is also a constant in $i$. Hence we have:
 \begin{align}
 \nonumber N=\mathbb{E}(T) \geq \frac{(d+1+\frac{d^2}{4})p_{i,i+1}}{2d\sum\limits_{k=1}^m k^2p_k}=\\
 \nonumber d\frac{p_{i,i+1}}{8\sum\limits_{k=1}^m k^2p_k}+
 \frac{p_{i,i+1}}{2\sum\limits_{k=1}^m k^2p_k}+
 \frac{p_{i,i+1}}{2d\sum\limits_{k=1}^m k^2p_k}\\
 \nonumber \geq \frac{(d+4+\frac{4}{d})p_{i,i+1}}{8m^2\sum\limits_{k=1}^m p_k}
 \geq \frac{(d+4+\frac{4}{d})p_{i,i+1}}{4m^2(1-p_0)}\\
 \geq \frac{(d+4+\frac{4}{d})p_{i,i+1}}{4m^2\eps(2-\eps)}
 \end{align}
using the fact that for two identical clocks playing the game, $2\sum\limits_{k=1}^m p_k=1-p_0$, and $$1-p_0 \leq 1- p_{ii}^A p_{kk}^B\leq1-(1-\eps)^2=\eps(2-\eps)$$ where $i$ and $k$ are positions of players A and B. Hence to maximise lower bound one needs $m^2\eps$ to be small, which again is minimised for a ladder clock (where $m=1$ with small $\eps$. 
    \end{proof}
    
    \section{Perfect stochastic clocks in the AT Game}
    Let A be a perfect ladder clock  and B be any stochastic clock B, and assume that  A is at least one position ahead of  B at the start. The difference equation obtained when the perfect clock plays the AT Game against  B is of the same form as Eq. \ref{eq:recurrence}, but with different probability distribution. Namely, one obtains:
              \begin{equation}
              \begin{split}
             &\mathbb{E}(S|Q^0=z)=D_z= p_{1}(D_{z+1}+1)+p_2(D_{z+2}+1)\\
             &+...+p_m(D_{z+m}+1)+ p_0(D_z+1)+p_{-1}(D_{z-1}+1)\\
             &+p_{-2}(D_{z-2}+1)+...+p_{-m}(D_{z-m}+1)=\\
             &q_{0}(D_{z}+1)+q_1(D_{z+1}+1)\\
             &+...+q_{-m}(D_{z+m+1}+1)+ q_2(D_{z-1}+1)\\
             &+...+q_{m}(D_{z-m+1}+1)
             \end{split}
             \end{equation}
where $q_k,\,k\in\lbrace -m,...,m \rbrace$ is the probability that clock B moves for $k$ states in one step.

However, note that in the case of different clocks playing, the obtained difference equation is not symmetric. Hence one can check that the particular solution as before (quadratic function in $z$) will not work this time. The homogenous solution is still $D_z^{\mathrm{hom}}=A+Bz$ (due to the fact that $\sum\limits_{k=i}^{i+m+1} p_{k,i}=1$).

  Note that once we have a probability distribution of a specific clock we can easily obtain average number of alternate ticks, since $\mathbb{E}(Y)=d$ in this case (due to the property of the perfect ladder clock).
  If perfect ladder clock would play against ordinary ladder clock, the solution of the differential equation becomes $D_z=D_0-\frac{z}{1-\delta},\,z\geq 0$. Note that since we assume that the perfect clock has an advantage at the beginning, it is not possible that the ladder clock overtakes it, hence $D_z=0$ for $z<0$ always. Boundary conditions for the lower bound then becomes $D_{d+1}=0$, which leads to $D_z^{\mathrm{lower}}=\frac{d+1-z}{1-\delta}$ and for upper bound $D_{2d+1}=0$, hence $D_z^{\mathrm{upper}}=\frac{2d+1-z}{1-\delta}$. Note that results are minimised for starting condition $z=0$. Hence
  \begin{equation}
  \frac{d+1}{1-\delta}\leq D_z^{\mathrm{ideal\, vs\, ladder}}\leq\frac{2d+1}{1-\delta}
  \end{equation}
    
\bibliographystyle{unsrt}

\fontsize{3pt}{1pt}\selectfont
\bibliography{refInflation}

       \end{document}